\documentclass[preprint,authoryear,11pt]{elsarticle}

\usepackage[T1]{fontenc}
\usepackage[utf8]{inputenc}
\usepackage[british]{babel}
\usepackage{lmodern}
\usepackage{microtype}
\usepackage{amsmath,amssymb,amsthm}
\usepackage{bm}
\usepackage{mathtools}
\usepackage{graphicx}
\usepackage{booktabs}
\usepackage{tikz}
\usepackage{enumitem}
\usepackage[hidelinks]{hyperref}
\usepackage{array}
\theoremstyle{plain}
\newtheorem{theorem}{Theorem}
\newtheorem{proposition}[theorem]{Proposition}
\newtheorem{lemma}[theorem]{Lemma}

\theoremstyle{definition}
\newtheorem{definition}[theorem]{Definition}
\theoremstyle{remark}
\newtheorem*{remark}{Remark}

\newcommand{\E}{\mathbb{E}}
\newcommand{\R}{\mathbb{R}}
\newcommand{\relent}{D}
\newcommand{\rdfun}{R}
\newcommand{\one}{\mathbf{1}}

\journal{Journal of Economic Theory}

\begin{document}

\begin{frontmatter}

\title{Choice at Finite Capacity: The Bounded Agent as an Information Channel and the Recovery of Walrasian Demand}

\author{Avishek Bhandari}
\address{School of Humanities, Social Sciences and Management, Indian Institute of Technology Bhubaneswar, India}
\ead{avishekb@iitbbs.ac.in}

\begin{abstract}
Textbook microeconomics gives the consumer unlimited reach, and from that rationality descend the slope of demand, the symmetry of substitution, and the integrability of a demand system. This paper treats the consumer instead as a finite-capacity information channel that compresses its economy to the resolution its bit budget allows and acts on that compression. Its choice is then a law, not a point: a prior tilted by the exponential of scaled value, whose inverse temperature prices information. Classical demand is recovered exactly as the infinite-capacity limit, a strict generalisation away from it, with pure habit at the zero-capacity corner. The central result is a compressed Slutsky matrix $S=-\beta\lambda\Sigma$, a scalar multiple of the covariance of the consumer's own choice law: symmetric at every finite capacity, negative semidefinite whenever the budget binds from above, and tending to the Walrasian substitution matrix as capacity grows without bound. These properties are inherited from the exponential form a capacity constraint imposes, not from rationality, so the compensated law of demand belongs to the budget and the channel. The construction is substrate-neutral, covering a machine running a capacity-limited policy; a worked two-good quadratic consumer carries every quantity in closed form.
\end{abstract}


\end{frontmatter}

\section{Introduction}
\label{sec:intro}

Economics begins with the choosing agent, and it inherited a particular agent in a particular form. From the marginalist turn through its axiomatic completion the consumer was taken to be a maximiser: a holder of a fixed, complete, transitive preference who, faced with prices and a wealth budget, selects the single bundle that stands highest in that preference among the affordable ones \citep{mascolell1995,debreu1959}. On that one assumption the science built the theory of demand, and built it tightly. The downward slope of a compensated demand curve, the symmetry of the Slutsky matrix, the integrability conditions that certify a demand system as having issued from a well-behaved preference: each is a theorem about the maximiser, and each fails without it. Rationality is not an ornament of the classical consumer but the load-bearing axiom, the premise from which the observable content of demand theory is deduced.

This paper keeps the deductive tightness and changes the premise. The consumer is modelled not as a reasoner of unlimited reach but as an information channel of finite capacity: a device that cannot resolve the full state of its economy and so compresses that state to the resolution a budget of bits affords, acting on the compressed image rather than on the world. The premise is not a friction added to an otherwise rational core. It is the primitive from which behaviour is derived, and the standard maximiser is recovered from it as a limit. A consumer who cannot track every price, a firm that cannot model every rival, a learning algorithm that cannot store every feature: each is a channel whose rate, the number of distinctions it can draw per decision, is finite and scarce, and each must decide how to spend that rate before it can decide anything else \citep{simon1955,sims2003}.

The change of premise has an exact consequence for the form of choice. An agent that cannot resolve the state finely enough to enact a single best response does not choose a point; it chooses a distribution. The optimal distribution is not posited but derived, as the solution of the agent's information-constrained problem: maximise expected value subject to a ceiling on the bits processed, or, in the equivalent priced form, maximise expected value net of the information spent in departing from a default way of acting. The solution is the Gibbs law
\begin{equation}
q^{\ast}(a)\ \propto\ p_0(a)\,\exp\!\big(\beta\,U(a)\big),
\label{eq:gibbs0}
\end{equation}
the exponential tilt of a prior $p_0$ by the value $U$ of action, at an inverse temperature $\beta$ that is the shadow price of information. This is the decision rule of rational inattention and of the free-energy and stochastic-control traditions, and we take it as the agent's primitive \citep{sims2010,matejkamckay2015,ortegabraun2013}. The parameter $\beta$ indexes capacity at its two ends. As $\beta\to\infty$ the law concentrates on the value-maximising action and the agent becomes the unconstrained optimiser of the canon; as $\beta\to0$ it collapses onto the prior $p_0$ and the agent becomes pure habit.

The central result of this paper is that classical demand theory is recovered exactly as the infinite-capacity limit of the bounded consumer, and is a strict generalisation away from it. This is a \emph{second nesting}, the agent-level counterpart of the same idealisation read at the level of the economy, where the competitive, rational-expectations equilibrium would be the zero-entropy-rate rest state of an adaptive information source. The two nestings are the two faces of one infinite-resource idealisation. Equilibrium is the economy at zero entropy rate, the state in which novelty has fallen to nothing; rationality is the agent at infinite capacity, the state in which the price of thought has fallen to nothing. Neither limit is the primitive: the perfect optimiser, like the still economy, is a limiting case, the frictionless corner of a theory whose subject is friction.

The technical heart of the recovery is a compressed Slutsky decomposition (Theorem~\ref{thm:slutsky}). The substitution matrix of the bounded consumer is
\begin{equation}
S\ =\ -\,\beta\lambda\,\Sigma,
\label{eq:S0}
\end{equation}
the negative of the inverse temperature $\beta$, times the marginal value of wealth $\lambda$, times the covariance $\Sigma$ of the consumer's own choice law. Three features distinguish this object from the existing study of substitution-matrix structure under bounded rationality, in which the Slutsky matrix of an observed demand is examined through the size of its violation of symmetry and of the compensated law \citep{aguiarserrano2017}. First, $S$ is \emph{symmetric and negative semidefinite at every finite capacity}, not only in the limit, because it is a scalar multiple of a covariance matrix; the two properties that the classical theory reads as the signature of utility maximisation are present in the bounded consumer as properties of the second moment of its behaviour. Second, the symmetry and the sign \emph{do not use the rationality of the consumer}: the derivation needs only that demand is the Gibbs law of some value over the budget set, so the compensated law of demand is inherited from the exponential form a capacity constraint imposes and survives the replacement of the consumer's preference by an arbitrary field (Proposition~\ref{prop:beckerlaw}). Third, the infinite-capacity limit of \eqref{eq:S0} is exactly the Walrasian, Hicksian substitution matrix (Theorem~\ref{thm:slutsky}, final clause): the rescaled covariance $\beta\Sigma$ tends to the inverse curvature of the value at the maximiser, and $S$ tends to the negative-definite matrix that classical theory derives from the second-order conditions of constrained maximisation \citep{slutsky1915,hicks1939}.

Two further consequences follow from the channel reading and are developed below. The consumer's budget acquires a second dimension, a constraint on the rate it can process that is as binding as the wealth constraint and carries its own price, the inverse capacity; the marginal value of wealth and the price of attention act on behaviour through the one law \eqref{eq:gibbs0}. And the construction is indifferent to the substrate of the agent. To call something a channel is to say nothing about what it is made of, only about what it does, and a machine running a capacity-limited policy, an information-bottleneck representation or a rate-distortion reinforcement learner, is to the theory the same object as the human consumer, a finite channel set at a capacity its designer chose \citep{tishbypereirabialek1999,arumugamvanroy2021}.

\paragraph{Related work.} The agent's optimisation is the rational-inattention problem of \citet{sims2003,sims2010} and its discrete-choice solution \citep{matejkamckay2015,caplindean2015}; the stochastic choice law is the random-utility logit of \citet{mcfadden1974}; the rate-distortion frontier is the classical object of \citet{shannon1959,berger1971,coverthomas2006}. The classical demand theory the construction nests is that of \citet{slutsky1915,hicks1939,samuelson1938,houthakker1950,afriat1967}, with the rationality-free reading of the compensated law due in spirit to \citet{becker1962} and set against the bounded-rationality Slutsky-norm programme of \citet{aguiarserrano2017}. Section~\ref{sec:relation} gives the full positioning, including the free-energy, granular-aggregation, and machine-learning connections \citep{ortegabraun2013,gabaix2011,tishbypereirabialek1999}.

\section{The bounded consumer}
\label{sec:setup}

We fix the agent as a channel, recall the rate-distortion frontier that is its budget, and derive its decision rule. The measure-theoretic and operator-theoretic apparatus is introduced only to the extent the present construction requires.

\subsection{The agent as a rate-distortion channel}

An agent stands between its economy and its actions as a conduit of limited throughput. The economy presents a stream of observables richer than the agent can register; the agent forms from that stream an internal representation under a bound on how much information the representation may carry; and it acts on the representation, because the world in full detail never reaches it. The bound is the agent's defining parameter.

\begin{definition}[The agent channel and its capacity]
\label{def:channel}
Let $X$ be the state the agent must act upon and $\hat X$ its internal reproduction, produced by a channel, a conditional law $p(\hat x\mid x)$. The channel is subject to a capacity constraint
\[
I(X;\hat X)\ \le\ C,
\]
on the mutual information, in bits per use, that the reproduction carries about the state. Given a distortion $d(x,\hat x)\ge0$ measuring the cost of reproducing $x$ as $\hat x$, the least information compatible with an expected distortion within a tolerance $D$ is the \emph{rate-distortion function}
\[
\rdfun(D)\ =\ \min_{\,p(\hat x\mid x)\,:\;\E[d(X,\hat X)]\le D}\ I(X;\hat X).
\]
A \emph{bounded agent} is a channel operating on this frontier, trading the rate $\rdfun$ it spends against the fidelity $D$ with which it must act.
\end{definition}

The function $\rdfun(D)$ is convex and decreasing, and these two properties are the whole economics of being a finite channel \citep{shannon1959,berger1971,coverthomas2006}. It is decreasing because tolerating more error always lowers the information one must acquire; it is convex because the easy reductions come first, so each further increment of fidelity costs more rate than the last. The agent therefore faces a trade-off with diminishing returns, the shape the economist meets everywhere, drawn in a new pair of axes: not goods against goods, but information against error. The canonical case fixes the intuition. For a Gaussian state of variance $\sigma^2$ under squared-error distortion the frontier is
\begin{equation}
\rdfun(D)\ =\ \tfrac12\log\frac{\sigma^2}{D}\quad(0\le D\le\sigma^2),
\qquad \rdfun(D)=0\quad(D>\sigma^2),
\label{eq:gaussianrd}
\end{equation}
so each halving of the tolerated error costs a fixed further half-bit, and a more volatile state costs more rate to track to any given fidelity \citep[Theorem~10.3.2]{coverthomas2006}. When the state is a vector of independent Gaussian modes the optimal channel allocates its rate across them by \emph{reverse water-filling}, spending bits first on the modes of largest variance and describing by no bits at all every mode whose variance lies below a common water level $w$, the level falling as the capacity rises \citep[Theorem~10.3.3]{coverthomas2006}. The agent's first economic decision, prior to and determining all the others, is where to sit on this frontier: how much of its economy to resolve.

\subsection{The free-energy choice}

The constrained problem, maximise value subject to a ceiling on information, has a dual that is easier to solve and more revealing to read. Rather than impose a hard ceiling on the rate, charge the agent a price for it. The constrained and priced problems are two views of one convex programme, related as a constraint to its Lagrangian, and the price is the multiplier on the rate. When the agent faces a single decision against a fixed state, the rate it spends is the relative entropy of its chosen law from its default; when it faces many states, the rate is the mutual information between state and action, and the default is then not an arbitrary prior but the agent's own unconditional law of action, recovered as part of the solution \citep{matejkamckay2015}. In either reading the agent's problem is to choose a law of action that maximises expected value net of the relative-entropy distance it must travel from its default, and the solution is a single recurring object.

\begin{proposition}[Bounded-rational choice as a free-energy optimum]
\label{prop:freeenergy}
Let an agent choose a probability law $p$ over a set of actions to maximise the free energy
\[
J(p)\ =\ \E_p[U]\ -\ \tfrac1\beta\,\relent(p\,\|\,p_0),
\]
where $U(a)$ is the value of action $a$, $p_0$ a default policy with $p_0(a)>0$, and $\beta>0$ prices the information spent in departing from $p_0$. Then $J$ is strictly concave and has the unique maximiser
\[
p^{\ast}(a)\ =\ \frac1Z\,p_0(a)\,e^{\beta U(a)},
\qquad Z=\sum_{a}p_0(a)\,e^{\beta U(a)},
\]
with optimal value $J(p^{\ast})=\tfrac1\beta\log Z$. As $\beta\to\infty$ the law $p^{\ast}$ concentrates on $\arg\max_a U(a)$; as $\beta\to0$ it returns to $p_0$.
\end{proposition}

\begin{proof}[Proof sketch]
The relative entropy $\relent(p\,\|\,p_0)$ is strictly convex in $p$, so $J$, a linear functional minus a strictly convex one, is strictly concave on the simplex and has a unique maximiser. A multiplier for the normalisation gives the stationarity condition $U(a)-\tfrac1\beta(\log(p(a)/p_0(a))+1)-\nu=0$, whence $p(a)\propto p_0(a)e^{\beta U(a)}$; normalising fixes $Z$ and substituting back collapses the value and entropy terms to $\tfrac1\beta\log Z$. As $\beta\to\infty$ the exponential weight concentrates the mass on the actions of maximal $U$, the prior breaking ties; as $\beta\to0$ the exponential factor tends to unity and $p^{\ast}\to p_0$. The full proof is in Appendix~\ref{app:proofs}.
\end{proof}

\begin{proposition}[The constrained and priced problems are dual]
\label{prop:duality}
Fix the value $U$ and default $p_0$. The information-constrained problem, maximise $\E_p[U]$ over laws $p$ subject to a rate ceiling $\relent(p\,\|\,p_0)\le R$, and the priced problem of Proposition~\ref{prop:freeenergy} are Lagrangian duals: the inverse temperature $1/\beta$ is the multiplier on the rate constraint, the optimal priced value is the free energy $J(p^{\ast})=\tfrac1\beta\log Z$, and as $\beta$ ranges over $(0,\infty)$ the priced optima trace the rate-distortion frontier $\rdfun$ of Definition~\ref{def:channel}, the spent rate $\relent(p^{\ast}\,\|\,p_0)$ falling as $\beta$ rises.
\end{proposition}

\begin{proof}[Proof sketch]
The objective $\E_p[U]$ is linear and the rate functional $\relent(p\,\|\,p_0)$ strictly convex in $p$, so the constrained problem is convex; for any $R$ between the rate of the value-greedy law and zero it has a Slater-interior point, so strong duality holds and a multiplier $\nu\ge0$ exists. The Lagrangian $\E_p[U]-\nu(\relent(p\,\|\,p_0)-R)$ is, up to the constant $\nu R$, the priced objective of Proposition~\ref{prop:freeenergy} with $\beta=1/\nu$, whose maximiser is the Gibbs law and whose value is $\tfrac1\beta\log Z$. Sweeping $R$, equivalently $\beta$, traces the attained value against the spent rate, the convex decreasing frontier $\rdfun$; the many-state reading replaces $\relent(p\,\|\,p_0)$ by the mutual information $I(X;\hat X)$ and $p_0$ by the unconditional law of action, with the same duality \citep{coverthomas2006,berger1971}.
\end{proof}

\begin{remark}
Proposition~\ref{prop:freeenergy} is stated for a finite action set and a proper default; the same variational characterisation holds for a continuous action space against a $\sigma$-finite reference measure, the maximiser being the Gibbs density $dq^{\ast}/d\nu\propto p_0\,e^{\beta U}$ by the Gibbs variational principle, with the flat prior of the worked example the Lebesgue reference. This continuous statement is the one used in Definition~\ref{def:cdemand} and below.
\end{remark}

The maximiser is a Gibbs law, the exponential tilt of the prior by the value of action, and the optimal value is a free energy, the log of a partition function divided by the inverse temperature $\beta$. The names are borrowed from thermodynamics, and the borrowing is one of variational form and not of any conserved economic substance or literal temperature: the agent that trades expected value against a relative-entropy cost solves the same variational problem a physical system solves when it minimises free energy at a fixed temperature, with $\beta$ the inverse temperature and $-U$ the energy \citep{ortegabraun2013}. The same object governs the agent that steers a dynamic environment by paying the relative-entropy cost of deflecting it from its uncontrolled course, and the agent that plans by treating control as inference; in each the optimal policy is a Gibbs tilt and the optimal cost a free energy \citep{todorov2009,kappen2012}.

\subsection{Rationality and habit as the two extremes of capacity}

The limits in Proposition~\ref{prop:freeenergy} are the structural content of the agent, and they place the maximiser of the textbooks inside the theory rather than outside it. The parameter $\beta$ is capacity expressed as a price, large when information is cheap and the channel wide, small when information is dear and the channel narrow. At the upper extreme the relative-entropy cost vanishes from the problem, the Gibbs law collapses onto the actions of maximal value, and the agent becomes the perfect optimiser, choosing the best response with certainty and paying nothing to compute it. At the lower extreme the agent can afford no departure from its prior and simply enacts its default, the creature of habit. Between them lies every real agent, human and artificial, buying with its finite capacity a graded improvement over habit that never reaches the counsel of perfection.

\begin{definition}[The bounded agent]
\label{p:agent}
An agent is an information channel of finite capacity. Its choice is the free-energy optimum of Proposition~\ref{prop:freeenergy}, a Gibbs law over actions proportional to a prior weight times the exponential of scaled value, whose concentration is governed by an inverse-temperature parameter $\beta$ that indexes the capacity. The unconstrained maximiser of the classical account is this agent at $\beta\to\infty$; pure habit is this agent at $\beta\to0$.
\end{definition}

Definition~\ref{p:agent} is the agent-level statement of the same recovery read for the economy as a whole, and the parallel is the point. There the competitive equilibrium is recovered as the rest state of an information source, the configuration at which the entropy rate falls to zero and the economy produces no novelty: a true and important state of the theory, but a limiting one, the still point of a system whose nature is to move. Here the rational maximiser is recovered as the infinite-capacity limit of the agent, the configuration at which the price of information falls to zero and the agent makes no error of compression: a true and important state, and equally limiting, the frictionless corner of a theory whose subject is friction. Equilibrium is the economy at zero entropy rate; rationality is the agent at infinite capacity. Neither is the primitive; each classical object is a limiting case of the more general one.

\subsection{The machine is the same channel}

Because the definition of a channel appeals to nothing but information and its cost, the agent of this paper need not be human, and the point is sharper here than for a believer, since the rate-distortion channel is an explicit object much of contemporary machine learning is built upon. The information bottleneck, the principle that an optimal internal representation minimises the bits it retains about its input while preserving the bits relevant to its task, is the rate-distortion problem of Definition~\ref{def:channel} with the distortion learned from the task rather than fixed in advance \citep{tishbypereirabialek1999}. A deep network is, on this reading, a cascade of such channels, each layer compressing its input toward a point on the same frontier, and the principle is realised operationally in learners trained to hold an explicit price on the information their representations carry \citep{tishbyzaslavsky2015,alemi2017}. The correspondence reaches the agents that learn by acting: a reinforcement learner can be made to solve a rate-distortion function at each step to decide what about its environment is worth learning at all, so the very target of its learning is set by the same constrained problem the economic agent solves \citep{arumugamvanroy2021}, and an agent that selects actions under a cost on the information its policy carries reproduces the stochastic, habit-anchored, capacity-limited behaviour described throughout, by the same equations \citep{laigershman2021}. The economic decision-maker as a finite channel, the deep network compressing its input, and the reinforcement learner deciding what to learn are one object at different budgets. The unification is of the modelling primitive and its optimisation, and the paper claims no more than that: an artificial agent entering the economy is not a new kind of participant demanding a new theory but the same channel the theory already names, set at a capacity its designers chose.

\section{The compressed consumer and the recovery of Walrasian demand}
\label{sec:consumer}

We now specialise the channel to the consumer and read what changes. The classical consumer chooses a point; the compressed consumer chooses a distribution, because an agent of finite rate cannot resolve the state finely enough to enact a single best response.

\begin{definition}[Compressed demand]
\label{def:cdemand}
Let a consumer face prices $p\in\R^L_{++}$ and wealth $w>0$, with budget set $B(p,w)=\{x\in\R^L_{+}:p\cdot x\le w\}$, and let $U(x)$ be the value it places on the bundle $x$. The \emph{compressed demand} is the law over bundles
\begin{multline*}
q^{\ast}(x\mid p,w)\ =\ \frac1Z\,p_0(x)\,e^{\beta U(x)}\,\one\{x\in B(p,w)\},\\
Z=\int_{B(p,w)}p_0(x)\,e^{\beta U(x)}\,dx,
\end{multline*}
the free-energy optimum of Proposition~\ref{prop:freeenergy} restricted to the budget set, with $p_0$ the consumer's prior over bundles and $\beta>0$ its capacity expressed as the inverse price of information. The \emph{mean demand} is $\bar x(p,w)=\E_{q^{\ast}}[x]$.
\end{definition}

The compressed demand is a genuine distribution, not a point with noise added: its dispersion is the trace of a capacity constraint, and the mean demand is a summary of it. The first proposition fixes its relation to the classical demand it generalises.

\begin{proposition}[The law of demand in distribution, and its classical limit]
\label{prop:demandlimit}
The compressed demand of Definition~\ref{def:cdemand} is supported in the budget set, so budget feasibility holds with certainty; it is non-degenerate at every finite capacity $\beta$ for a prior of full support; and as $\beta\to\infty$ it concentrates on $\arg\max_{x\in B(p,w)}U(x)$, the Walrasian demand of the consumer whose preference $U$ represents. The classical demand function is therefore the infinite-capacity limit of the compressed demand, the consumer-level instance of the agent limit of Definition~\ref{p:agent}.
\end{proposition}

\begin{proof}[Proof sketch]
Feasibility is the indicator in Definition~\ref{def:cdemand}, which annihilates the density off $B(p,w)$. Non-degeneracy holds because for finite $\beta$ the density is strictly positive wherever $p_0$ is on the relative interior of the compact set $B(p,w)$, so no bundle carries all the mass. The limit restricts Proposition~\ref{prop:freeenergy} to $B(p,w)$: as $\beta\to\infty$ the Gibbs weight concentrates on the maximisers of $U$ over the support, which over $B(p,w)$ are by definition the Walrasian demand, the prior breaking ties and compactness securing attainment.
\end{proof}

The distributional reading already breaks one classical equivalence. Revealed-preference theory asks when observed choices can be rationalised by a single fixed utility; a compressed consumer's realised choices need not be so rationalisable, since across repeated budgets a non-degenerate demand chooses, with positive probability, bundles standing in the cyclic relation the strong axiom forbids, and a single such cycle is incompatible by Afriat's theorem with rationalisation by any utility \citep{afriat1967,samuelson1938,houthakker1950}. The probability of such a violation falls to zero as $\beta\to\infty$, where the choices converge to the Walrasian demand of a fixed utility and are rationalisable. Rationalisability is therefore not a precondition of the analysis but its infinite-capacity limit.

\subsection{The two budgets and the compressed Slutsky decomposition}

The deepest change the channel reading makes to the consumer is to its budget. The classical consumer faces one constraint, the wealth budget $p\cdot x\le w$, whose shadow price is the marginal value of wealth. The compressed consumer faces two. It must also pay for the information by which it resolves its situation, and the second constraint, a ceiling on the rate it can process, is as binding as the first and carries its own price, the inverse capacity $1/\beta$. The first prices goods in money; the second prices attention in bits. The two prices act through the one law, and their joint action is read most cleanly in the consumer's response to a change in price, the object the Slutsky equation organises in the classical theory. That equation has an exact counterpart.

For the price derivative we use the priced, soft-budget form of the compressed demand,
\begin{equation}
q^{\ast}(x)\ \propto\ p_0(x)\,\exp\!\big(\beta[\,U(x)-\lambda\,p\cdot x\,]\big),
\label{eq:priced}
\end{equation}
the Lagrangian dual of the hard-constraint demand of Definition~\ref{def:cdemand}, with $\lambda=\lambda(p,w)$ the multiplier holding the wealth budget binding in expectation, $\E_{q^{\ast}}[p\cdot x]=w$. The priced form is the smooth surrogate of the hard-constraint demand: its support is all of $\R^L_+$ and the budget holds in expectation rather than with certainty, and it is the surrogate on which the price derivative is taken, because the hard budget indicator is not differentiable in $p$. We take this priced law as the object of the differential demand analysis, so that the mean demand $\bar x(p,w)=\E_{q^{\ast}}[x]$ and its covariance are read from it; the priced and hard-constraint laws share neither mean nor covariance at a finite capacity, the two coinciding in the infinite-capacity limit (Lemma~\ref{lem:softconc}), where the hard-constraint law of Definition~\ref{def:cdemand} records that feasibility holds with certainty. The multiplier is strictly positive, $\lambda>0$, exactly when the wealth budget binds from above, that is when the unconstrained Gibbs mean is unaffordable; this holds whenever the value is locally non-satiated, and for a satiable value whenever wealth falls short of the cost of its ideal point. Throughout the price derivative the default $p_0$ is held fixed, as in the single-decision reading in which the rate spent is the relative entropy $\relent(q^{\ast}\,\|\,p_0)$ from the agent's unconditional law; the determination of that law as a fixed point across the distribution of states is a separate matter that does not enter the within-state response to price.

\begin{lemma}[Concentration of the priced law]
\label{lem:softconc}
Fix $(p,w)$ and let $\lambda^{\ast}=\lim_{\beta\to\infty}\lambda(p,w)$. Suppose $U(x)-\lambda^{\ast}\,p\cdot x$ attains its maximum over $\R^L_+$ at a point $x^{\ast}$ interior to $\R^L_+$ with $-\nabla^2U(x^{\ast})$ positive definite. Then as $\beta\to\infty$ the priced law \eqref{eq:priced} concentrates weakly on $x^{\ast}$, and $x^{\ast}$ is the maximiser of $U$ over the budget hyperplane $\{x:p\cdot x=w\}$.
\end{lemma}

\begin{proof}[Proof sketch]
The priced law is the Gibbs tilt of $p_0$ by $\beta(U-\lambda\,p\cdot x)$; by Proposition~\ref{prop:freeenergy} in its continuous form, applied with reference measure $p_0$, it concentrates as $\beta\to\infty$ on the maximisers of $U-\lambda^{\ast}p\cdot x$, here the single interior point $x^{\ast}$. The interior stationarity condition $\nabla U(x^{\ast})=\lambda^{\ast}p$ is the first-order condition of $\max\{U(x):p\cdot x=w\}$ with multiplier $\lambda^{\ast}$, and the in-expectation budget $\E_{q^{\ast}}[p\cdot x]=w$ forces $p\cdot x^{\ast}=w$ in the limit, identifying $x^{\ast}$ with the budget-constrained maximiser.
\end{proof}

\begin{theorem}[The compressed Slutsky decomposition]
\label{thm:slutsky}
Let a consumer have compressed demand in the priced form \eqref{eq:priced}, and write $\Sigma=\mathrm{Cov}_{q^{\ast}}(x)$ for the covariance of its choice law and $\bar x(p,w)=\E_{q^{\ast}}[x]$ for its mean demand. Then the response of mean demand to price decomposes as
\[
\frac{\partial\bar x_i}{\partial p_j}
\ =\ S_{ij}\ -\ \beta\,\frac{\partial\lambda}{\partial p_j}\,\mathrm{Cov}_{q^{\ast}}(x_i,\,p\cdot x),
\qquad
S_{ij}\ :=\ -\,\beta\lambda\,\mathrm{Cov}_{q^{\ast}}(x_i,x_j),
\]
in which, when the wealth budget binds from above so that $\lambda>0$, the substitution matrix $S=-\beta\lambda\,\Sigma$ is symmetric and negative semidefinite at every finite capacity, and the second term is the wealth adjustment. In particular the compensated own-price effect is $S_{ii}=-\beta\lambda\,\mathrm{Var}_{q^{\ast}}(x_i)\le0$. As $\beta\to\infty$ the law $q^{\ast}$ concentrates, by Lemma~\ref{lem:softconc}, on the Walrasian demand of the consumer, and the rescaled covariance $\beta\Sigma$ tends to the inverse curvature $(-\nabla^2U)^{-1}$ of the value at the maximiser, so $S=-\beta\lambda\,\Sigma$ tends to the negative-definite matrix $-\lambda(-\nabla^2U)^{-1}$; the classical Slutsky substitution matrix is the restriction of this limit to the budget hyperplane, the rank-one singularity along $p$ supplied by the wealth-adjustment term, so the full decomposition reproduces the classical Slutsky equation \citep{slutsky1915,hicks1939} in the limit.
\end{theorem}

\begin{proof}[Proof sketch]
The law \eqref{eq:priced} is an exponential family in the sufficient statistic $x$ with natural parameter $\eta=-\beta\lambda\,p$ and base measure $p_0(x)e^{\beta U(x)}$. For an exponential family the mean is the gradient of the log-partition function in the natural parameter and its derivative is the covariance, $\partial\E_{q^{\ast}}[x_i]/\partial\eta_k=\mathrm{Cov}_{q^{\ast}}(x_i,x_k)$. Price enters only through $\eta=-\beta\lambda p$ with $\lambda=\lambda(p,w)$, so the chain rule gives $\partial\bar x_i/\partial p_j=\sum_k\mathrm{Cov}_{q^{\ast}}(x_i,x_k)(-\beta\lambda\delta_{kj}-\beta p_k\,\partial\lambda/\partial p_j)$, separating into $S_{ij}=-\beta\lambda\,\mathrm{Cov}_{q^{\ast}}(x_i,x_j)$ and the wealth-adjustment term. Symmetry of $S$ is the symmetry of the covariance; for the sign, $a^{\top}Sa=-\beta\lambda\,\mathrm{Var}_{q^{\ast}}(a\cdot x)\le0$ for any $a$, since $\beta\lambda>0$. The limit is the concentration of $q^{\ast}$ established in Proposition~\ref{prop:demandlimit} together with the Laplace expansion $\beta\Sigma\to(-\nabla^2U(x^{\ast}))^{-1}$. The full proof is in Appendix~\ref{app:proofs}.
\end{proof}

That the compressed substitution matrix is symmetric and negative semidefinite is the heart of the matter. In the classical theory those two properties are the integrability conditions, the signature that a demand system has come from a well-behaved preference; they are guaranteed only for a maximiser of a concave utility. Here they hold for the compressed consumer as properties of a covariance matrix, the covariance of its own choice law, and the proof used no property of $U$ beyond measurability. The integrability conditions are inherited not from the consumer's rationality but from the exponential form its bounded rate imposes on its behaviour, and they survive the replacement of its preference by an arbitrary field. This is the precise sense in which the recovery of demand theory does not depend on the premise the classical theory could not do without.

\begin{proposition}[The compensated law of demand from the budget, not from rationality]
\label{prop:beckerlaw}
The compensated law of demand holds for the compressed consumer whenever the wealth budget binds with a strictly positive shadow price, irrespective of the rationality of its value. For any value field $U$ and prior $p_0$ for which $\lambda>0$, the own-price substitution effect of Theorem~\ref{thm:slutsky} satisfies $S_{ii}=-\beta\lambda\,\mathrm{Var}_{q^{\ast}}(x_i)\le0$, with equality only for a good the consumer never varies; the slope is therefore a property of the budget and the channel, not of the preference, and it survives the failure of completeness, transitivity, and utility-representability. When instead the value is satiable and its ideal point affordable, $\lambda\le0$ and the compensated own-price slope is correspondingly non-negative, a signed prediction the channel reading delivers rather than an exception to it. By a separate argument, for a population choosing by any law supported on the budget set and fixed independently of the price change, a compensated rise in a good's own price contracts the opportunity set and cannot raise the population-mean consumption of that good \citep{becker1962}; this population statement does not subsume the single compressed consumer, whose law $q^{\ast}$ itself moves with the price through $\eta=-\beta\lambda p$.
\end{proposition}

\begin{proof}[Proof sketch]
The inequality $S_{ii}\le0$ is the diagonal of the negative semidefinite matrix of Theorem~\ref{thm:slutsky}, whose proof used no property of $U$ beyond measurability and so holds verbatim when $U$ is the realised value of any field; equality requires zero variance, a good held fixed. The population statement is the budget-set monotonicity of \citet{becker1962}: a compensated own-price rise rotates the opportunity set about the chosen bundle, shrinking the portion in which the dearer good is the more abundant, and the mean of a family of laws supported on the set and fixed independently of the price cannot rise in that coordinate as the set contracts.
\end{proof}

The classical theory derived the downward slope of compensated demand from the consumer's rationality, through the negativity of a substitution matrix that only a maximiser of a well-behaved utility was guaranteed to possess. The compressed theory derives it from the budget. The opportunity set contracts when a price rises, and the contraction moves the mass of any choice law, rational or not, off the good that has become dearer; the maximisation was never doing the work, in signing the compensated slope, that the constraint was doing. The uncompensated slope is another matter, retaining the wealth term of Theorem~\ref{thm:slutsky} and unsigned, so a compressed consumer can display a Giffen good exactly as a rational one can. What the budget secures without rationality is the compensated law alone.

\subsection{Utility under adaptation}

One further classical fixity is loosened by the channel reading, and it bears on the interpretation of $U$. The classical consumer maximises a utility given in advance and held fixed, the discipline's methodological charter having made the fixity a principle that assigns all variation in conduct to prices and none to taste \citep{stiglerbecker1977}. The compressed consumer's value is more naturally read as a learned object, formed by experience and equal to the fixed utility of the textbooks only when the learning has converged. Holding tastes fixed is then the assertion that the consumer's learning has already settled, a legitimate idealisation for an environment stationary long enough, and the general case is the consumer whose environment changes faster than its value can settle. The results above are stated for a fixed $U$, the converged case; the path-dependent case, in which the value carries the history of the consumer's experience, is left to separate work and does not bear on the nesting, which is a statement about the limit in $\beta$ at a fixed value.

\section{A worked example: a two-good quadratic consumer}
\label{sec:worked}

The recovery of the previous section is a statement about a limit in capacity. This section makes it concrete in a single consumer small enough to carry every quantity in closed form, yet rich enough to exhibit the compressed demand, the compressed Slutsky matrix and its symmetry and sign, the infinite-capacity recovery of the Walrasian substitution matrix, the zero-capacity collapse to habit, and the Gaussian rate-distortion frontier, all at once. The vehicle is a two-good quadratic consumer. Every number below has been computed in closed form and verified numerically.

\subsection*{The consumer and its mean demand}

A consumer chooses a bundle $x=(x_1,x_2)\in\R^2_+$ and places on it the quadratic value
\begin{equation}
\label{eq:value}
U(x)\ =\ a\cdot x-\tfrac12\,x^{\top}Q\,x,
\qquad
a=\begin{pmatrix}4\\3\end{pmatrix},\quad
Q=\begin{pmatrix}2&\tfrac12\\[2pt]\tfrac12&1\end{pmatrix},
\end{equation}
with $Q$ symmetric and positive definite, so $U$ is strictly concave and satiable; the curvature matrix $-\nabla^2U=Q$ has eigenvalues $0.7929$ and $2.2071$, both positive. Prices are $p=(1,1)$ and wealth is $w=3$. The consumer is a channel of capacity $\beta$, so by Definition~\ref{def:cdemand} it does not solve the constrained maximisation directly; it draws its bundle from the compressed demand. Take a flat prior $p_0$, the maximum-entropy prior on the plane. In the priced form \eqref{eq:priced} the argument of the exponential is, up to an additive constant,
\begin{equation}
\label{eq:complete}
U(x)-\lambda\,p\cdot x
=-\tfrac12\big(x-\mu\big)^{\top}Q\big(x-\mu\big)+\text{const},
\qquad \mu=Q^{-1}(a-\lambda p),
\end{equation}
so the compressed demand is the Gaussian
\begin{equation}
\label{eq:gauss}
q^{\ast}(x)\ =\ \mathcal N\!\big(\mu,\ \Sigma\big),
\qquad
\mu=Q^{-1}(a-\lambda p),
\qquad
\Sigma=(\beta Q)^{-1}.
\end{equation}
Its mean is the constrained best response $\bar x=\mu$, and the multiplier $\lambda$ is set to hold the budget binding in expectation, $p\cdot\bar x=w$. Since $\bar x=Q^{-1}(a-\lambda p)$, the budget condition is linear in $\lambda$ and gives
\begin{multline}
\label{eq:lambda}
\lambda=\frac{p^{\top}Q^{-1}a-w}{p^{\top}Q^{-1}p}=\frac{3.7143-3}{1.1429}=0.625,\\
\bar x=\begin{pmatrix}1.25\\1.75\end{pmatrix},
\qquad p\cdot\bar x=3=w.
\end{multline}
The unconstrained bliss point $Q^{-1}a=(1.4286,2.2857)$ costs $p\cdot Q^{-1}a=3.7143>w$, so the budget binds and the multiplier is strictly positive, as required. The mean demand $(1.25,1.75)$ is the value the infinite-capacity limit must reproduce; the variance around it is the cost of finite capacity.

\subsection*{(a) The compressed Slutsky matrix: symmetry and sign}

By Theorem~\ref{thm:slutsky} the substitution matrix is $S=-\beta\lambda\,\Sigma$ with $\Sigma=(\beta Q)^{-1}$. The factor of $\beta$ cancels, an exact feature of the Gaussian case that makes the matrix available in closed form:
\begin{multline}
\label{eq:Sclosed}
S=-\beta\lambda\,(\beta Q)^{-1}=-\lambda\,Q^{-1}\\
=-0.625\begin{pmatrix}\phantom{-}0.5714&-0.2857\\-0.2857&\phantom{-}1.1429\end{pmatrix}
=\begin{pmatrix}-0.3571&\phantom{-}0.1786\\\phantom{-}0.1786&-0.7143\end{pmatrix}.
\end{multline}
The matrix is manifestly symmetric, its off-diagonal entries equal at $0.1786$, the symmetry inherited from the covariance and not assumed. Its eigenvalues are $-0.7883$ and $-0.2832$, both strictly negative, so $S$ is negative definite, hence negative semidefinite, as Theorem~\ref{thm:slutsky} requires. The compensated own-price effects are the diagonal entries, $S_{11}=-\beta\lambda\,\mathrm{Var}(x_1)=-0.3571$ and $S_{22}=-\beta\lambda\,\mathrm{Var}(x_2)=-0.7143$, both negative: a compensated rise in a good's own price lowers its mean consumption, the compensated law of demand, here a property of the second moment of the choice law. None of these signs used any property of $U$ beyond its curvature; replacing $U$ by any value field with the same second-order behaviour about the chosen bundle leaves the symmetry and the sign of $S$ intact, the content of Proposition~\ref{prop:beckerlaw}.

That the matrix in \eqref{eq:Sclosed} is independent of $\beta$ is a special exact feature of the Gaussian, in which the mean and the multiplier do not move with capacity; what does move with capacity is the dispersion of the choice law about that mean, reported in Table~\ref{tab:converge}. In a non-Gaussian consumer the mean demand and the multiplier vary with $\beta$, and the matrix approaches its limit \eqref{eq:Slim} only as $\beta\to\infty$; the quadratic affords the closed form in which the limit is already attained.

\subsection*{(b) The infinite-capacity limit recovers Walrasian demand}

As $\beta\to\infty$ the variance $\Sigma=(\beta Q)^{-1}\to0$, so by Proposition~\ref{prop:demandlimit} the compressed demand \eqref{eq:gauss} concentrates, weakly, on the point mass at the constrained best response,
\begin{equation}
\label{eq:agentlimit}
q^{\ast}\ \xrightarrow{\ \beta\to\infty\ }\ \delta_{\,\bar x},
\qquad \bar x=(1.25,1.75),
\end{equation}
which is exactly the Walrasian demand of the consumer computed directly at \eqref{eq:lambda}: the bundle that maximises $U$ on the budget set. The substitution matrix tends, by Theorem~\ref{thm:slutsky}, to
\begin{equation}
\label{eq:Slim}
S\ \xrightarrow{\ \beta\to\infty\ }\ -\lambda\,(-\nabla^2U)^{-1}=-\lambda\,Q^{-1}
=\begin{pmatrix}-0.3571&\phantom{-}0.1786\\\phantom{-}0.1786&-0.7143\end{pmatrix},
\end{equation}
the Hicksian substitution matrix derived in the classical theory from the second-order conditions of constrained maximisation \citep{hicks1939}. The classical demand theory, the maximising bundle and its negative-definite substitution matrix alike, is recovered term for term as the infinite-capacity corner of the compressed consumer, and not by stipulation: the limit removed the only source of randomness, the compression noise of variance $\Sigma$, and what is left is the deterministic choice the canon would have written down.

\subsection*{(c) The zero-capacity limit collapses to habit}

The other extreme is the consumer of pure habit. To exhibit it we give the consumer a proper prior, $p_0=\mathcal N(m_0,I)$ with $m_0=(1,1)$, in place of the flat prior, so that as capacity vanishes the choice has somewhere to fall back to. The compressed demand is then Gaussian with precision $\beta Q+I$ and a mean that interpolates between the prior centre and the value-driven best response. Holding the wealth shadow price at its reference value $\lambda=0.625$ to isolate the capacity effect, the mean demand traces
\begin{equation}
\label{eq:priorcollapse}
\bar x(\beta)\ \xrightarrow{\ \beta\to0\ }\ m_0=(1,1),
\end{equation}
falling from $(1.256,1.598)$ at $\beta=4$ through $(1.132,1.162)$ at $\beta=\tfrac14$ to the prior centre $(1,1)$ as $\beta\to0$ (Table~\ref{tab:converge}). At zero capacity the consumer can afford no departure from its prior and simply enacts its default: it is the creature of habit of Definition~\ref{p:agent}, condition $\beta\to0$. The two limits of the one parameter are therefore the two consumers economics has always known, the perfect optimiser and the pure creature of habit, recovered as the endpoints of a single capacity axis with every real consumer in between.

\subsection*{(d) The rate-distortion frontier the consumer sits on}

The capacity the consumer spends is priced on the Gaussian rate-distortion frontier \eqref{eq:gaussianrd}. Reading the relevant state as a standardised mode of unit variance, $\sigma^2=1$, the least rate compatible with a reproduction error $D$ is $\rdfun(D)=\tfrac12\log_2(1/D)$ bits, the convex decreasing curve of Table~\ref{tab:rd}: each halving of the tolerated error from $D=\tfrac12$ downward costs a fixed further half-bit, from $0.5$ bits at $D=\tfrac12$ to $2.0$ bits at $D=\tfrac1{16}$. This is the budget line of the second constraint of the consumer's problem, drawn in bits rather than money: it says how much capacity the consumer must buy to resolve its situation to a given fidelity, and the inverse temperature $\beta$ of \eqref{eq:gauss} is the price at which it buys, the shadow value of the rate. The figure summarises the two limits and the frontier together.

\begin{table}[t]
\centering
\footnotesize
\setlength{\tabcolsep}{2pt}
\begin{tabular}{@{}r *{4}{>{\centering\arraybackslash}p{0.218\textwidth}}@{}}
\toprule
$\beta$ & mean $\bar x$ (flat prior) & dispersion $\sqrt{\mathrm{Var}(x_1)}$ &
$S$ eigenvalues & mean $\bar x$ ($\beta\to0$ prior) \\
\midrule
$4$    & $(1.2500,1.7500)$ & $0.3780$ & $(-0.7883,-0.2832)$ & $(1.2561,1.5976)$\\
$16$   & $(1.2500,1.7500)$ & $0.1890$ & $(-0.7883,-0.2832)$ & $(1.2535,1.7042)$\\
$64$   & $(1.2500,1.7500)$ & $0.0945$ & $(-0.7883,-0.2832)$ & $(1.2511,1.7379)$\\
$\tfrac14$ & $(1.2500,1.7500)$ & $1.5119$ & $(-0.7883,-0.2832)$ & $(1.1324,1.1618)$\\
\midrule
$\beta\to\infty$ & $(1.2500,1.7500)$ & $0$ & $(-0.7883,-0.2832)$ & Walrasian $\delta_{\bar x}$\\
$\beta\to0$      & undefined & $\infty$ & undefined & habit $m_0=(1,1)$\\
\bottomrule
\end{tabular}
\caption{The capacity axis of the two-good quadratic consumer ($a=(4,3)$,
$Q=\big(\begin{smallmatrix}2&.5\\.5&1\end{smallmatrix}\big)$, $p=(1,1)$, $w=3$).
With a flat prior the mean demand is the Walrasian bundle $(1.25,1.75)$ at every
finite capacity, since the Gaussian mean is the constrained best response, and the
substitution matrix $S=-\lambda Q^{-1}$ has the fixed eigenvalues $(-0.7883,-0.2832)$,
symmetric and negative definite throughout; the dispersion $\sqrt{\mathrm{Var}(x_1)}=
\sqrt{(Q^{-1})_{11}/\beta}=\sqrt{4/(7\beta)}$ falls to zero as $\beta\to\infty$ (the
infinite-capacity, Walrasian corner) and grows without bound as $\beta\to0$. The last column, computed with a
proper prior $p_0=\mathcal N(m_0,I)$, $m_0=(1,1)$, and the wealth price held at its
reference value to isolate the capacity effect, shows the mean falling to the prior
centre as $\beta\to0$, the collapse to habit.}
\label{tab:converge}
\end{table}

\begin{table}[t]
\centering
\small
\begin{tabular}{rcc}
\toprule
distortion $D$ & rate $\rdfun(D)$ (bits) & rate $\rdfun(D)$ (nats)\\
\midrule
$1.0000$ & $0.0000$ & $0.0000$\\
$0.5000$ & $0.5000$ & $0.3466$\\
$0.2500$ & $1.0000$ & $0.6931$\\
$0.1250$ & $1.5000$ & $1.0397$\\
$0.0625$ & $2.0000$ & $1.3863$\\
\bottomrule
\end{tabular}
\caption{The Gaussian rate-distortion frontier $\rdfun(D)=\tfrac12\log(\sigma^2/D)$
of \eqref{eq:gaussianrd} at $\sigma^2=1$, the consumer's information budget line. The
function is convex and decreasing: each halving of the tolerated error $D$ from
$\tfrac12$ downward costs a fixed further half-bit. The inverse temperature $\beta$ of
the compressed demand \eqref{eq:gauss} is the price at which the consumer buys position
on this curve.}
\label{tab:rd}
\end{table}

\begin{figure}[t]
\centering
\begin{tikzpicture}[>=stealth,scale=1.0]
% axes
\draw[->,thick] (0,0) -- (6.4,0) node[right]{$D$ (distortion)};
\draw[->,thick] (0,0) -- (0,4.3) node[above]{$\rdfun(D)$ (bits)};
% rate-distortion curve R = 1/2 log2(1/D), scaled: x = 6*D, y = 2*R for D in (0.06,1]
% points (D, R): (1,0)(.5,.5)(.25,1)(.125,1.5)(.0625,2)
\draw[thick,blue,domain=0.0625:1,samples=80,smooth,variable=\d]
  plot ({6*\d},{2*(0.5*log2(1/\d))});
\node[blue] at (5.2,3.4) {$\rdfun(D)=\tfrac12\log_2\tfrac1D$};
% mark the tabulated points
\foreach \d/\r in {1/0,0.5/0.5,0.25/1,0.125/1.5,0.0625/2}
  \fill[blue] ({6*\d},{2*\r}) circle (1.6pt);
\node[below] at ({6*1},0) {\scriptsize $1$};
\node[below] at ({6*0.5},0) {\scriptsize $\tfrac12$};
\node[below] at ({6*0.25},0) {\scriptsize $\tfrac14$};
\node[left] at (0,{2*0.5}) {\scriptsize $0.5$};
\node[left] at (0,{2*1}) {\scriptsize $1$};
\node[left] at (0,{2*2}) {\scriptsize $2$};
% annotate the two ends
\draw[<-,gray] ({6*0.0625},{2*2}) -- ++(0.2,0.5)
  node[right,gray,font=\scriptsize,align=left]{$D\downarrow0$:\\ $\beta\to\infty$, Walrasian};
\draw[<-,gray] ({6*1},{2*0}) -- ++(-0.1,0.6)
  node[above,gray,font=\scriptsize,align=center]{$D=\sigma^2$:\\ $\beta\to0$, habit};
\end{tikzpicture}
\caption{The consumer's information budget line. The Gaussian rate-distortion frontier
$\rdfun(D)=\tfrac12\log(\sigma^2/D)$ at $\sigma^2=1$ is convex and decreasing; the marked
points are those of Table~\ref{tab:rd}. The high-fidelity end $D\downarrow0$, where the
consumer buys unbounded rate, is the infinite-capacity corner at which Walrasian demand
is recovered; the no-resolution end $D=\sigma^2$, where the consumer buys no bits, is the
zero-capacity corner of pure habit.}
\label{fig:rd}
\end{figure}
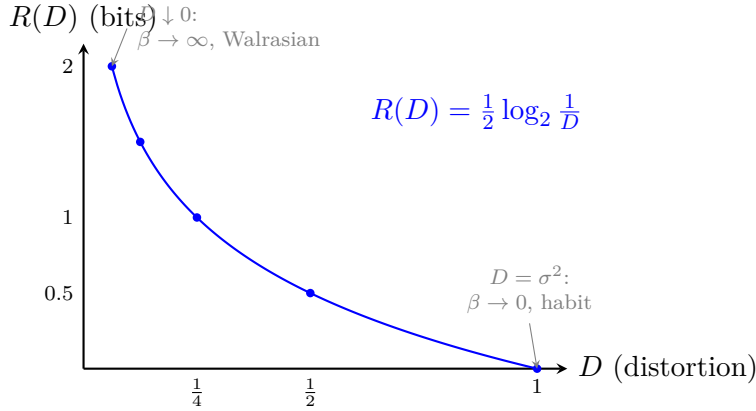

\begin{remark}
The example is deliberately Gaussian, which makes the compressed demand a Gaussian, the substitution matrix available as $-\lambda Q^{-1}$ in closed form, and the mean demand independent of the capacity, so that the infinite-capacity Walrasian substitution matrix is already attained at every finite $\beta$. None of the structural content depends on that special feature. Replacing the quadratic value by any strictly concave value with a well-separated maximiser on the budget set leaves the construction intact, provided the budget binds from above so that $\lambda>0$: the compressed demand is non-degenerate at finite capacity, its substitution matrix $S=-\beta\lambda\,\Sigma$ stays symmetric and negative semidefinite by Theorem~\ref{thm:slutsky}, and the limit $\beta\to\infty$ delivers the Walrasian demand and the Hicksian substitution matrix of that value, the rescaled covariance $\beta\Sigma$ converging to the inverse curvature at the maximiser. Unlike the Gaussian, whose mean demand is independent of the capacity, a non-quadratic value has a mean demand and a multiplier that move with the capacity and reach the Walrasian point only in the limit. What the quadratic buys is only that every number can be written down.
\end{remark}

\section{Related literature}
\label{sec:relation}

The bounded consumer of this paper stands at the meeting of several literatures, and the contribution is a reading rather than an estimator: that each of these programmes is a face of a single nesting, in which the classical maximiser is the infinite-capacity corner of a finite channel.

The agent's optimisation is the rational-inattention problem of \citet{sims2003,sims2010}, in which an agent processing information at a finite Shannon rate chooses as if maximising attainable value net of a mutual-information cost; the literature is surveyed in \citet{mackowiak2023}. The Gibbs law of Definition~\ref{p:agent} is the solution of exactly that problem, and the multinomial-logit form it takes in a discrete choice is the result of \citet{matejkamckay2015}, who derive the prior-weighted logit, its weights the agent's own unconditional choice probabilities rather than free coefficients, from the same variational principle; that this behaviour has firm observable content, recoverable from the way choices vary with the state, is the contribution of \citet{caplindean2015}. What the present reading adds is not the agent's optimisation, which it inherits, but the demand-theoretic consequence: that the entire integrability apparatus of classical demand, the symmetry and negativity of the substitution matrix, is recovered from the exponential form the capacity constraint imposes, and recovered without the rationality of the consumer (Theorem~\ref{thm:slutsky}, Proposition~\ref{prop:beckerlaw}).

The stochastic form of the choice law connects the construction to the random-utility tradition. The Gibbs law \eqref{eq:gibbs0} over a finite menu is precisely the choice probability when the systematic value is perturbed by independent extreme-value noise, the multinomial logit of \citet{mcfadden1974}, the prior entering as an additive log-offset. The bounded consumer is therefore a random-utility agent, genuinely stochastic at every finite capacity and deterministic only in the limit; and where the random-utility model is regular, the channel reading locates that regularity as the logit corner of the wider space of stochastic-choice rules, in which the intransitive and incomplete fields of \citet{thurstone1927} and \citet{luce1959} also live. The cost the agent pays for information is by now the standard primitive of the inattention literature, adopted here as such while recording that its claim to be the uniquely correct cost of attention is contested, and that costs tied to the perceptual difficulty of distinctions rather than to information content alone remain a live alternative the theory should not foreclose \citep{mackowiak2023}.

The rationality-free reading of the compensated law is in the spirit of \citet{becker1962}, who showed that a compensated own-price rise lowers mean consumption for a population choosing irrationally on a budget set, the slope a property of the constraint and not of the agent. Theorem~\ref{thm:slutsky} sharpens that observation to the full substitution matrix of a single channel and ties it to a specific behavioural primitive, the Gibbs law, of which the rationalisable consumer is the infinite-capacity limit. Where \citet{aguiarserrano2017} measure the distance of an observed Slutsky matrix from rationality, sizing the violation of symmetry and of the compensated law, the channel substitution matrix here is symmetric and negative semidefinite by construction, a rescaled covariance, so the compensated block sits at zero violation by the exponential form whenever the budget binds from above, and what their norm would measure is carried by the wealth term and the dispersion of the choice law. The same step bears on the methodological charter of \citet{stiglerbecker1977}: holding tastes fixed is recovered as the assertion that the consumer's value has finished adapting, a special case rather than a postulate. Where the consumer population is heterogeneous in scale, the idiosyncratic does not wash out, and the largest members transmit their attention and their errors to the aggregate by the granularity logic of \citet{gabaix2011}, so the aggregate of compressed choice is not a representative household but a structured process. That the standard revealed-preference tests can fail for a compressed consumer at the individual level \citep{samuelson1938,houthakker1950,afriat1967} while the aggregate remains legible is the demand-theoretic face of the same point.

The substrate-neutral reading places the construction beside the machine-learning literature on bounded representation. The information bottleneck of \citet{tishbypereirabialek1999} is the rate-distortion problem of Definition~\ref{def:channel} with a learned distortion, realised operationally in \citet{tishbyzaslavsky2015} and \citet{alemi2017}; the rate-distortion reading of reinforcement learning is \citet{arumugamvanroy2021}, and the capacity-limited policy of \citet{laigershman2021} reproduces the same stochastic, habit-anchored behaviour. The free-energy and stochastic-control readings of the Gibbs law, in \citet{ortegabraun2013}, \citet{todorov2009}, and \citet{kappen2012}, and the abstraction-forming reading of \citet{genewein2015}, supply the dynamic and cognitive faces of the one functional. The synthesis is the point: the consumer of demand theory, the inattentive agent of macroeconomics, the random-utility agent of discrete choice, and the capacity-limited learner of machine learning are one object at different budgets, and the classical maximiser is the corner each of them approaches as its capacity grows without bound.

\section{Conclusion}
\label{sec:conclusion}

The consumer of the textbooks chooses a point and is rational by axiom, and on that axiom the observable content of demand theory rests. This paper has put a different consumer beneath the same theory: a finite-capacity information channel whose choice is a law, the Gibbs distribution that solves its information-constrained problem, with an inverse temperature that is the shadow price of information. The classical consumer is not refuted but located. It is the channel at infinite capacity, the frictionless corner at which the price of thought has fallen to zero, recovered there exactly, the agent-level counterpart of the recovery of competitive equilibrium as a zero-novelty rest state; at the opposite corner, zero capacity, the same consumer is pure habit, acting on its prior alone; and between the two lies every consumer that has ever chosen.

The technical content of the recovery is the compressed Slutsky matrix $S=-\beta\lambda\Sigma$, symmetric and negative semidefinite at every finite capacity because it is a multiple of a covariance, and tending to the Hicksian substitution matrix as capacity diverges. The two properties classical theory reads as the signature of rational maximisation are present in the bounded consumer as properties of the second moment of its own behaviour, and they hold for any value the budget binds from above, so the compensated law of demand is a property of the budget and the channel and not of rationality. The maximisation, in signing the compensated slope, was never doing the work the constraint was doing. The construction extends without change to a machine running a capacity-limited policy, an artificial agent being not a new kind of participant but the same channel at a capacity its designers chose. What the bounded consumer carries away from the infinite-capacity corner, a demand that is a distribution, a budget that prices attention, and behaviour that need not be rationalisable, is the content the classical primitive, which knows only the maximising point, cannot express. The maximiser is the still, frictionless, infinitely sharp corner of a chooser whose nature is to compress, to err, and to economise on thought.

\appendix
\section{Proofs}
\label{app:proofs}

This appendix gives the proofs in fuller form than the sketches of the main text.

\subsection*{Proof of Proposition~\ref{prop:freeenergy}}
Fix a finite action set $\{a_1,\dots,a_m\}$ and write $U_j=U(a_j)$, $p_{0,j}=p_0(a_j)>0$. The objective is $J(p)=\sum_j p_jU_j-\tfrac1\beta\sum_j p_j\log(p_j/p_{0,j})$ over the simplex $\{p\ge0,\sum_j p_j=1\}$. The map $p\mapsto\sum_j p_j\log(p_j/p_{0,j})$ is strictly convex, being a sum of strictly convex perspective terms, and $\sum_j p_jU_j$ is linear, so $J$ is strictly concave and attains a unique maximiser, interior because the entropy term has infinite slope at the boundary of the simplex. A multiplier $\nu$ for the constraint $\sum_j p_j=1$ gives the stationarity condition
\[
U_j-\tfrac1\beta\Big(\log\frac{p_j}{p_{0,j}}+1\Big)-\nu=0,
\]
whence $\log(p_j/p_{0,j})=\beta U_j-\beta\nu-1$ and $p_j\propto p_{0,j}e^{\beta U_j}$; normalising fixes the partition function $Z=\sum_j p_{0,j}e^{\beta U_j}$ and returns $p^{\ast}$. Substituting $p^{\ast}$ into $J$ and using $\log(p^{\ast}_j/p_{0,j})=\beta U_j-\log Z$ together with $\sum_j p^{\ast}_j=1$, the value and entropy contributions cancel to leave $J(p^{\ast})=\tfrac1\beta\log Z$. For the limits, write $p^{\ast}_j\propto p_{0,j}e^{\beta U_j}$ and let $U^{\star}=\max_j U_j$ with maximising set $A^{\star}$: for $a_j\notin A^{\star}$ the ratio $p^{\ast}_j/\sum_{k\in A^{\star}}p^{\ast}_k=(p_{0,j}/\sum_{k\in A^{\star}}p_{0,k})e^{-\beta(U^{\star}-U_j)}\to0$ as $\beta\to\infty$, so the mass concentrates on $A^{\star}$ with the prior breaking ties within it; as $\beta\to0$ the factor $e^{\beta U_j}\to1$ and $p^{\ast}_j\to p_{0,j}$. \qed

\subsection*{Proof of Proposition~\ref{prop:demandlimit}}
Feasibility is immediate, the indicator $\one\{x\in B(p,w)\}$ in Definition~\ref{def:cdemand} setting the density to zero off the budget set, which therefore carries the whole mass. For non-degeneracy, the budget set $B(p,w)$ is compact with non-empty relative interior, and for finite $\beta$ the density $Z^{-1}p_0(x)e^{\beta U(x)}$ is continuous and strictly positive on that interior wherever $p_0$ is, so the law assigns positive probability to every relatively open subset and no single bundle carries all the mass. For the limit, $B(p,w)$ is compact and $U$ continuous, so $\arg\max_{x\in B(p,w)}U(x)$ is non-empty; the argument of Proposition~\ref{prop:freeenergy}, applied with the support restricted to $B(p,w)$, gives weak concentration of $q^{\ast}$ on that argmax set as $\beta\to\infty$, the prior breaking ties. When $U$ represents the consumer's preference, the maximiser of $U$ over $B(p,w)$ is by definition the Walrasian demand, so the compressed demand concentrates on it. \qed

\subsection*{Proof of Theorem~\ref{thm:slutsky}}
Write the priced law \eqref{eq:priced} as $q^{\ast}(x)=Z(\eta)^{-1}h(x)\exp(\eta^{\top}x)$ with base measure $h(x)=p_0(x)e^{\beta U(x)}$, the default $p_0$ held fixed under the price change, and natural parameter $\eta=-\beta\lambda\,p$. This is an exponential family in the sufficient statistic $x$. For such a family the mean is the gradient of the log-partition function, $\bar x=\nabla_\eta\log Z(\eta)$, and the Hessian of $\log Z$ is the covariance, so
\[
\frac{\partial\bar x_i}{\partial\eta_k}=\frac{\partial^2\log Z}{\partial\eta_i\,\partial\eta_k}=\mathrm{Cov}_{q^{\ast}}(x_i,x_k).
\]
Price enters the law only through $\eta=-\beta\lambda(p,w)\,p$, with $\lambda(p,w)$ determined by the in-expectation budget $\E_{q^{\ast}}[p\cdot x]=w$. Differentiating $\eta_k=-\beta\lambda p_k$ in $p_j$ gives $\partial\eta_k/\partial p_j=-\beta\lambda\,\delta_{kj}-\beta p_k\,\partial\lambda/\partial p_j$, and the chain rule yields
\begin{multline*}
\frac{\partial\bar x_i}{\partial p_j}
=\sum_k\mathrm{Cov}_{q^{\ast}}(x_i,x_k)\Big(-\beta\lambda\,\delta_{kj}-\beta p_k\frac{\partial\lambda}{\partial p_j}\Big)\\
=\underbrace{-\beta\lambda\,\mathrm{Cov}_{q^{\ast}}(x_i,x_j)}_{S_{ij}}
-\beta\frac{\partial\lambda}{\partial p_j}\,\mathrm{Cov}_{q^{\ast}}(x_i,\,p\cdot x),
\end{multline*}
using $\sum_k p_k\,\mathrm{Cov}_{q^{\ast}}(x_i,x_k)=\mathrm{Cov}_{q^{\ast}}(x_i,p\cdot x)$ in the second term. Thus $S=-\beta\lambda\,\Sigma$ with $\Sigma=\mathrm{Cov}_{q^{\ast}}(x)$. Symmetry of $S$ is the symmetry of $\Sigma$; for the sign, when the budget binds from above so that $\lambda>0$, any $a\in\R^L$ gives $a^{\top}Sa=-\beta\lambda\,a^{\top}\Sigma a=-\beta\lambda\,\mathrm{Var}_{q^{\ast}}(a\cdot x)\le0$ since $\beta\lambda>0$ and a variance is non-negative, and the own-price case $a=e_i$ gives $S_{ii}=-\beta\lambda\,\mathrm{Var}_{q^{\ast}}(x_i)\le0$. The derivation used no property of $U$ beyond the measurability that makes $h(x)$ a base measure. Differentiability of $\lambda(p,w)$ in $p$, used in the decomposition, follows from the implicit-function theorem applied to $\E_{q^{\ast}}[p\cdot x]-w=0$, whose derivative in $\lambda$ is $-\beta\,\mathrm{Var}_{q^{\ast}}(p\cdot x)<0$ by the full support of the priced law. For the limit, assume the constrained maximiser $x^{\ast}$ is interior to $\R^L_+$ with $-\nabla^2U(x^{\ast})$ positive definite, and note $\lambda(p,w)\to\lambda^{\ast}$ as $\beta\to\infty$, so the linear term $\lambda\,p\cdot x$ has vanishing Hessian and the curvature of the smooth exponent $\beta(U-\lambda\,p\cdot x)$ at the maximiser is $-\beta\nabla^2U(x^{\ast})$. Lemma~\ref{lem:softconc} gives $q^{\ast}\to\delta_{x^{\ast}}$, and the Laplace expansion gives $q^{\ast}$ asymptotically Gaussian with covariance $\Sigma\sim(\beta(-\nabla^2U(x^{\ast})))^{-1}$, so $\beta\Sigma\to(-\nabla^2U(x^{\ast}))^{-1}$ and $S=-\lambda(\beta\Sigma)\to-\lambda(-\nabla^2U(x^{\ast}))^{-1}$, negative definite since $-\nabla^2U(x^{\ast})$ is positive definite. The classical Slutsky substitution matrix is the restriction of this full matrix to the budget hyperplane, singular in the direction $p$, the rank-one wealth projection supplied by the second term of the decomposition. \qed

\subsection*{Proof of Proposition~\ref{prop:beckerlaw}}
The inequality $S_{ii}=-\beta\lambda\,\mathrm{Var}_{q^{\ast}}(x_i)\le0$ is the diagonal of the negative semidefinite matrix of Theorem~\ref{thm:slutsky}, whose proof used no property of $U$ beyond measurability; it therefore holds verbatim when $U$ is the realised value of an arbitrary preference field, in particular one that fails completeness, transitivity, or utility-representability, the only requirement being that the realised demand is the Gibbs law of some value over the budget set. Equality $S_{ii}=0$ requires $\mathrm{Var}_{q^{\ast}}(x_i)=0$, a good the consumer never varies. For the population statement, let the chosen bundles be distributed by a law $F$ supported on $B(p,w)$ and fixed independently of an own-price change in good $i$ that is wealth-compensated to hold the chosen reference bundle affordable. The compensated rise rotates the budget hyperplane about that bundle, and the new opportunity set $B(p',w')$ contains, relative to the old, fewer bundles with large $x_i$ and more with small $x_i$; integrating $x_i$ against a law supported on the contracting-in-$x_i$ set cannot raise its mean, so the population-mean consumption of good $i$ does not rise \citep{becker1962}. \qed

\bibliographystyle{elsarticle-harv}
\bibliography{refs}

\end{document}